\documentclass[11pt]{article}
\usepackage{amsmath}
\usepackage{amssymb}
\usepackage{amsfonts}

\usepackage{cite}

\usepackage{graphicx}
\usepackage{float}
\usepackage{longtable}
\usepackage[utf8]{inputenc}
\usepackage{hyperref}

\usepackage{algorithmicx}
\usepackage[ruled,vlined]{algorithm2e}
\SetAlgoCaptionSeparator{ }

\usepackage{fullpage}

\newcommand{\fref}[1]{Fig.~\ref{#1}}

\newcommand{\tref}[1]{Table~\ref{#1}}
\newcommand{\sref}[1]{Section~\ref{#1}}

\newenvironment{algo}[1][!htbp]
  {
   \begin{algorithm}[#1]%
  }{\end{algorithm}}

\newenvironment{proced}[1][!htbp]
  {
   \begin{algorithm}[#1]%
  }{\end{algorithm}}

\providecommand{\U}[1]{\protect\rule{.1in}{.1in}}
\newtheorem{theorem}{Theorem}

\newtheorem{lemma}{Lemma}

\newtheorem{problem}{Problem}
\newtheorem{proposition}{Proposition}

\newenvironment{proof}[1][Proof]{\textbf{#1.} }{\ \rule{0.5em}{0.5em}}

\begin{document}
\title{State Stabilization for Gate-Model Quantum Computers}
\author{Laszlo Gyongyosi\thanks{School of Electronics and Computer Science, University of Southampton, Southampton SO17 1BJ, U.K., and Department of Networked Systems and Services, Budapest University of Technology and Economics, 1117 Budapest, Hungary, and MTA-BME Information Systems Research Group, Hungarian Academy of Sciences, 1051 Budapest, Hungary.}
\and Sandor Imre\thanks{Department of Networked Systems and Services, Budapest University of Technology and Economics, 1117 Budapest, Hungary.}}
\date{}

\maketitle
\begin{abstract}
Gate-model quantum computers can allow quantum computations in near-term implementations. The stabilization of an optimal quantum state of a quantum computer is a challenge, since it requires stable quantum evolutions via a precise calibration of the unitaries. Here, we propose a method for the stabilization of an optimal quantum state of a quantum computer through an arbitrary number of running sequences. The optimal state of the quantum computer is set to maximize an objective function of an arbitrary problem fed into the quantum computer. We also propose a procedure to classify the stabilized quantum states of the quantum computer into stability classes. The results are convenient for gate-model quantum computations and near-term quantum computers.
\end{abstract}

\section{Introduction}
\label{sec1}
Quantum computers can make possible quantum computations for efficient problem solving \cite{ref1,ref2,ref3,ref4,ref5,ref6,ref7,ref8,ref9,ref10,ref11,ref12,ref13,ref14,ref15,ref16,ref17,ref18,ref19}. Gate-based quantum computations represent a way to construct gate-model quantum computers. In a gate-model quantum computer architecture, computations are implemented via sequences of unitary operations \cite{ref9,ref10,ref11,ref12, ref19,ref20,ref21,ref22,refp1,refp2,refp3,refp4,refp5}. Gate-model quantum computers allow establishing experimental quantum computations in near-term architectures \cite{refpr,refha,aar,logic,tel,depth,ft,refa3,refa4,refa5,refa6,refa7,song}. Practical demonstrations of gate-model quantum computers have been already proposed \cite{ref1,ref2,ref3,ref4,ref5,ref6,ref7,ref8,ref9,ref10,ref11,ref12} and several physical-layer developments are currently in progress.

Finding a stable quantum state of a quantum computer is a challenge, since it requires precise unitaries that yield stable quantum evolutions in the quantum computer. The problem is further increased if the stable system state must be available for a pre-determined time or for a pre-determined number of running sequences. Particularly, the quantum state of a quantum computer subject to stabilization also coincides with the optimal quantum state. The optimal quantum state of a quantum computer maximizes a particular objective function of an arbitrary computational problem fed into the quantum computer. The problem therefore is to fix the quantum state of the quantum computer in the optimal state for an arbitrary number of running sequences that is determined by the actual environment or by the current problem. Another challenge connected to the problem of stabilization of the system state of a quantum computer is the classification of the sequences of the stabilized quantum states into stability-classes. Practically, a solution to these problems can be covered by an unsupervised learning method.

Here, we propose a method for the stabilization of an optimal quantum state of a quantum computer through an arbitrary number of running sequences. We define a solution that utilizes unsupervised learning algorithms to determine the stable quantum states of the quantum computer and to classify the stable quantum states into stability classes. The proposed results are useful for experimental gate-based quantum computations and near-term quantum computer architectures.

The novel contributions of our manuscript are as follows:
\begin{enumerate}
\item We propose a method for the stabilization of an optimal quantum state of a quantum computer through an arbitrary number of running sequences. 
\item We define a solution that utilizes unsupervised learning algorithms to determine the stable quantum states of the quantum computer.
\item We evaluate a solution to classify the stable system states into stability classes. 
\end{enumerate}
This paper is organized as follows. \sref{sec2} provides the problem statement. \sref{sec3} discusses the stabilization procedure of an optimal quantum state of a quantum computer. \sref{sec4} defines an unsupervised learning method to find the stable quantum states and the stability classes of the stabilized quantum states. In \sref{nume}, a numerical evaluation is proposed. Finally, \sref{sec5} concludes with the results. Supplemental information is included in the Appendix.

\section{Problem Statement}
\label{sec2}
Let $QG$ be the quantum gate structure of a gate-model quantum computer with a sequence of $L$ unitaries \cite{ref9,ref10,ref11,ref12} with an $n$-length input system $\left| \psi  \right\rangle $, 
\begin{equation} \label{in}
\left| \psi  \right\rangle =\sum\limits_{i=0}^{{{d}^{n}-1}}{{{\alpha }_{i}}\left| i \right\rangle },
\end{equation}
where $d$ is the dimension ($d$=2 for a qubit system), $\sum\nolimits_{i=0}^{{{d}^{n}}-1}{{{\left| {{\alpha }_{i}} \right|}^{2}}}=1$, and let 
\begin{equation} \label{1)} 
{| \vec{\theta }^{*} \rangle} =U_{L} \left(\theta _{L}^{*} \right)U_{L-1} \left(\theta _{L-1}^{*} \right)\ldots U_{1} \left(\theta _{1}^{*} \right) \left| \psi  \right\rangle  
\end{equation} 
be the optimal system state of the quantum computer that maximizes a particular objective function $f(\vec{\theta }^{*} )$,
\begin{equation}
f(\vec{\theta }^{*} )=\langle \vec{\theta }^{*} |C|\vec{\theta }^{*} \rangle  
\end{equation}
of an arbitrary problem fed into the quantum computer, where $C$ is the classical value of the objective function, while $\vec{\theta }^{*} $ is the gate parameter vector, 
\begin{equation} \label{2)} 
\vec{\theta }^{*} =\left[\theta _{1}^{*} ,\ldots ,\theta _{L}^{*} \right]^{T}  
\end{equation} 
that identifies the $L$ unitaries, $U_{1} \left(\theta _{1}^{*} \right),\ldots ,U_{1} \left(\theta _{L}^{*} \right)$, of the $QG$ quantum circuit of the quantum computer in the optimal state ${| \vec{\theta }^{*} \rangle} $, such that an $i$-th unitary, $U_{i} \left(\theta _{i}^{*} \right)$ is as \cite{ref10}
\begin{equation} \label{3)} 
U_{i} \left(\theta _{i}^{*} \right)=\exp \left(-i\theta _{i}^{*} P\right),                                                               
\end{equation} 
where $\theta _{i}^{*}$ is the gate parameter (real continuous variable) of unitary $U_{i}$, $P$ is a generalized Pauli operator formulated by the tensor product of Pauli operators $\left\{X,Y,Z\right\}$ \cite{ref10,ref11}. 

The aim is to stabilize the ${| \vec{\theta }^{*} \rangle} $ optimal state of the quantum computer through $R$ running sequences via unsupervised learning of the evolution of the unitaries in the quantum computer. 

The $R$ running sequences refers to $R$ input systems fed into the input of the quantum computer, such that in an $r$-th running sequence, $r=1,\ldots ,R$, an $r$-th input system, $\left| {{\psi }_{r}} \right\rangle $ (defined as in \eqref{in}), is evolved via the sequence of the $L$ uniaries of the quatum computer. The $R$ running sequences identify an input system, $\left| {{\psi }_{in}} \right\rangle $, formulated via $R$, $n$-length quantum systems, as
\begin{equation}
\left| {{\psi }_{in}} \right\rangle =\left| {{\psi }_{1}} \right\rangle \otimes \ldots \otimes \left| {{\psi }_{R}} \right\rangle ,
\end{equation}
where it is considered that the $R$ input systems are unentangled.

Let $\vec{\varphi }$ be the gate parameter vector associated with the stable system state ${\left| \vec{\varphi } \right\rangle} $, 
\begin{equation} \label{ZEqnNum600779} 
{\left| \vec{\varphi } \right\rangle} =U_{L} \left(\varphi _{L} \right)U_{L-1} \left(\varphi _{L-1} \right)\ldots U_{1} \left(\varphi _{1} \right) \left| \psi  \right\rangle 
\end{equation} 
as
\begin{equation} \label{5)} 
\vec{\varphi }=\left[\varphi _{1} ,\ldots ,\varphi _{L} \right]^{T} ,                                                             
\end{equation} 
where $\varphi _{i} \in \left[0,\pi \right]$ is the gate parameter of unitary $U_{i} $ in the stabilized system state ${\left| \vec{\varphi } \right\rangle} $, such that the objective function value is stabilized into
\begin{equation} \label{6)} 
f(\vec{\varphi })=\langle \vec{\varphi }|C|\vec{\varphi }\rangle =f(\vec{\theta }^{*}).                                                     
\end{equation} 
For the $R$ sequences of the quantum computer, we define matrices $\alpha $ and $\beta $ as 
\begin{equation} \label{ZEqnNum579558} 
\alpha =[\vec{\theta }_{1}^{*} ,\ldots ,\vec{\theta }_{R}^{*}],                                                                 
\end{equation} 
where $\vec{\theta }_{r}^{*} =[\theta _{r,1}^{*} ,\ldots ,\theta _{r,L}^{*}]^{T}$ identifies the quantum state ${| \vec{\theta }_{r}^{*} \rangle}$ of an $r$-th running sequence of the quantum computer, while 
\begin{equation} \label{ZEqnNum728789} 
\beta =\left[\vec{\varphi }_{1} ,\ldots ,\vec{\varphi }_{R} \right],                                                                
\end{equation} 
where $\vec{\varphi }_{r} =\left[\varphi _{r,1} ,\ldots ,\varphi _{r,L} \right]^{T} $, identifies the stabilized quantum state ${\left| \vec{\varphi }_{r}  \right\rangle} $ of an $r$-th sequence of the quantum computer.

The problem therefore is to find $\beta $ from $\alpha $ that stabilizes the ${| \vec{\theta }^{*} \rangle} $ optimal state of the quantum computer through $R$ sequences as
\begin{equation} \label{ZEqnNum481709} 
\beta =S^{T} \alpha , 
\end{equation} 
where $S$ is a stabilizer matrix, 
\begin{equation} \label{10)} 
S^{T} S=I,                                                                 
\end{equation} 
and $I$ is the identity matrix.

The problems to be solved are therefore summarized as follows. 
\begin{problem}
Find $S$ to construct $\beta $ \eqref{ZEqnNum728789} from $\alpha $ \eqref{ZEqnNum579558} to stabilize the quantum computer in ${| \vec{\theta }^{*} \rangle} $ via ${\left| \vec{\varphi } \right\rangle} $ for all running sequences.
\end{problem}
\begin{problem}
Describe the stability of $\beta $ via unsupervised learning of the stability levels of the ${\left| \vec{\varphi }_{r}  \right\rangle} $ quantum states of $\beta $.
\end{problem}

The resolutions of Problems 1 and 2 are proposed in Theorems 1 and 2. The solution framework ${\rm {\mathcal F}}$ is defined via a ${\rm {\mathcal P}}_{S} $ stabilization procedure with an embedded stabilization algorithm ${\rm {\mathcal A}}_{S} $ (see Theorem 1), and via an ${\rm {\mathcal A}}_{C} $ classification algorithm that characterizes the stability class of the results of  ${\rm {\mathcal P}}_{S} $ (see Theorem 2). \fref{fig1} depicts the system model.

\begin{center}
\begin{figure*}[!htbp]
\begin{center}
\includegraphics[angle = 0,width=0.9\linewidth]{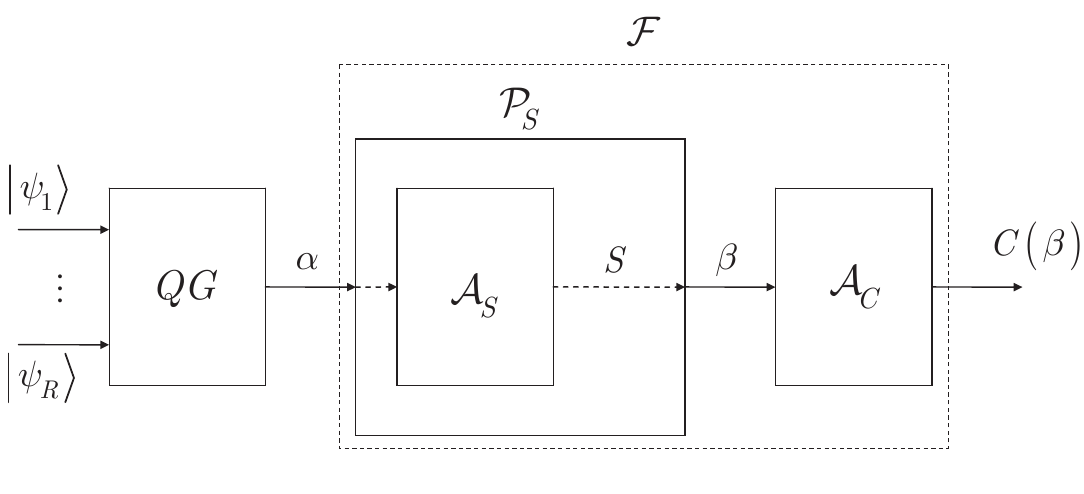}
\caption{The framework ${\rm {\mathcal F}}$ for the stabilization of the optimal state of the quantum computer and the stability-class determination. In the $R$ running sequences, $R$ input systems are fed into the input of the quantum computer, in an $r$-th running sequence, $r=1,\ldots ,R$, an $r$-th input system, $\left| {{\psi }_{r}} \right\rangle =\sum\nolimits_{i}{{{\alpha }_{i}}}\left| i \right\rangle $, is evolved via the sequence of the $L$ uniaries of the quatum computer. The $R$ running sequences identify an input system $\left| {{\psi }_{in}} \right\rangle =\left| {{\psi }_{1}} \right\rangle \otimes \ldots \otimes \left| {{\psi }_{R}} \right\rangle $ (considering that the $R$ input systems are unentangled). The $R$ running sequences of the $QG$ structure of the quantum computer produces $\alpha =[\vec{\theta }_{1}^{*} ,\ldots ,\vec{\theta }_{R}^{*}]$, where $\vec{\theta }_{r}^{*} =[\theta _{r,1}^{*} ,\ldots ,\theta _{r,L}^{*}]^{T}$. The ${\rm {\mathcal P}}_{S} $ stabilization procedure outputs $\beta =\left[\vec{\varphi }_{1} ,\ldots ,\vec{\varphi }_{R} \right]$, where $\vec{\varphi }_{r} =\left[\varphi _{r,1} ,\ldots ,\varphi _{r,L} \right]^{T} $, via an embedded stabilization algorithm ${\rm {\mathcal A}}_{S} $ that determines the $S$ stabilizer matrix. The $C\left(\beta \right)$ stability-level of the resulting $\beta $ is determined via a classification algorithm ${\rm {\mathcal A}}_{C} $. The ${\rm {\mathcal P}}_{S} $ and ${\rm {\mathcal A}}_{S} $ methods are realized as unsupervised learning.} 
 \label{fig1}
 \end{center}
\end{figure*}
\end{center}

\section{Stabilization of the Optimal State of the Quantum Computer}
\label{sec3}
\begin{theorem}
The $S$ matrix for the stabilization of the ${| \vec{\theta }^{*} \rangle} $ optimal state of the quantum computer via $\beta =S^{T} \alpha ,$ can be determined via the minimization of an objective function $F^{*} $.
\end{theorem}
\begin{proof}
For an $r$-th sequence of the quantum computer, define $\Delta (\vec{\theta }_{r}^{*})$ and $\Delta \left(\vec{\varphi }_{r} \right)$ as
\begin{equation} \label{11)} 
\Delta (\vec{\theta }_{r}^{*})=\vec{\theta }_{r}^{*} -\vec{\theta }_{r+1}^{*}  
\end{equation} 
and
\begin{equation} \label{12)} 
\Delta \left(\vec{\varphi }_{r} \right)=\vec{\varphi }_{r} -\vec{\varphi }_{r+1} ,                                                        
\end{equation} 
respectively. These vectors formulate $\Delta \alpha $ and $\Delta \beta $ as
\begin{equation} \label{ZEqnNum837426} 
\Delta \alpha =[\Delta (\vec{\theta }_{1}^{*}),\ldots ,\Delta (\vec{\theta }_{R-1}^{*})] 
\end{equation} 
and
\begin{equation} \label{ZEqnNum242571} 
\Delta \beta =\left[\Delta \left(\vec{\varphi }_{1} \right),\ldots ,\Delta \left(\vec{\varphi }_{R-1} \right)\right] ,                                                 
\end{equation} 
respectively. Then, using equations \eqref{ZEqnNum837426} and \eqref{ZEqnNum242571} for the $r=1,\ldots ,R-1$ sequences, let $\chi $ be a sum defined as
\begin{equation} \label{ZEqnNum337228} 
\begin{split}
   \chi &=\sum\limits_{r}^{R-1}{\left\| \Delta \left( {{{\vec{\varphi }}}_{r}} \right) \right\|_{2}^{2}}\\&=\text{Tr}\left( \Delta \beta {{\left( \Delta \beta  \right)}^{T}} \right) \\ 
 & =\text{Tr}\left( {{S}^{T}}\left( \Delta \alpha {{\left( \Delta \alpha  \right)}^{T}} \right)S \right),  
\end{split}
\end{equation} 
where $\left\| \cdot \right\| _{2}^{2} $  is the squared ${\rm L}2$-norm, ${\rm Tr}\left(\cdot \right)$ is the trace operator, $\Delta \alpha $ is as given in equation \eqref{ZEqnNum837426}, and $\Delta \beta $ is as in equation \eqref{ZEqnNum242571}. 

For the $r$-th and $s$-th sequences, $s>r$,  with $\Delta \left(\vec{\varphi }_{r} \right)$ and $\Delta \left(\vec{\varphi }_{s} \right)$, let $\gamma _{rs} $ be defined as
\begin{equation} \label{ZEqnNum418305} 
\gamma _{rs} =\omega _{rs} \left\| \Delta \left(\vec{\varphi }_{r} \right)-\Delta \left(\vec{\varphi }_{s} \right)\right\| _{2}^{2} ,                                                 
\end{equation} 
where $\omega _{rs} $ is a weight coefficient defined as
\begin{equation} \label{ZEqnNum241414} 
\omega _{rs} =\left\{\begin{array}{l} {\exp \left(-{\textstyle\frac{\left\| \Delta (\vec{\theta }_{r}^{*})-\Delta \left(\vec{\theta }_{s}^{*} \right)\right\| ^{2} }{\zeta }} \right),{\rm if\; }\left(s-r\right)\le \kappa } \\ {0,{\rm \; otherwise}} \end{array}\right. , 
\end{equation} 
where $\kappa $ and $\zeta $ are nonzero parameters.

A sum is defined for the $r=1,\ldots ,R-1$ sequences of the quantum computer as
\begin{equation} \label{ZEqnNum122889} 
\tau =\sum _{r}^{R-1}\sum _{s}^{R-1}\gamma _{rs}   . 
\end{equation} 
At a particular $S$ in equations \eqref{ZEqnNum337228} and \eqref{ZEqnNum122889}, the stabilization of the optimal state of the quantum computer through the $R$ sequences can be reformulated via an objective function $F^{*} $, subject to a minimization as
\begin{equation} \label{ZEqnNum168524} 
\begin{split}
   {{F}^{*}}&=\arg \underset{S}{\mathop{\min }}\,\left( \chi +c\tau  \right) \\ 
 & =\arg \underset{S}{\mathop{\min }}\,\left( \text{Tr}\left( {{S}^{T}}\left( \Delta \alpha {{\left( \Delta \alpha  \right)}^{T}} \right)S \right)+c\tau  \right),  
\end{split}
\end{equation} 
where $c$ is a regularization constant \cite{ref23,ref24}. The $F^{*} $ objective function therefore stabilizes the optimal state via the minimization of  $\chi $, while the term $c\tau $ achieves stabilization between the sequences.

Then, let $W$ be the weight matrix formulated via the coefficients \eqref{ZEqnNum241414} with $W_{rs} =\omega _{rs} $, and let $\eta $ be a diagonal matrix of the weight coefficients \eqref{ZEqnNum241414} with 
\begin{equation} \label{20)} 
\eta _{rr} =\sum _{s}\omega _{rs}  ,                                                                    
\end{equation} 
such that 
\begin{equation} \label{21)} 
\left(\Delta \beta \right)^{T} \eta \Delta \beta =I.                                                                
\end{equation} 
Using $W$ and $\eta $, the $F^{*} $ objective function in equation \eqref{ZEqnNum168524} can be rewritten as
\begin{equation} \label{ZEqnNum782750} 
\begin{split}
   {{F}^{*}}&=\arg \underset{S}{\mathop{\min }}\,\left( \tfrac{1}{\Omega }\left( \text{Tr}\left( \Delta \beta {{\left( \Delta \beta  \right)}^{T}} \right)+c\text{Tr}\left( \Delta \beta \left( \eta -W \right){{\left( \Delta \beta  \right)}^{T}} \right) \right) \right) \\ 
 & =\arg \underset{S}{\mathop{\min }}\,\left( \tfrac{1}{\Omega }\left( \text{Tr}\left( \Delta \beta \left( I+c\left( \eta -W \right) \right){{\left( \Delta \beta  \right)}^{T}} \right) \right) \right) \\ 
 & =\arg \underset{S}{\mathop{\min }}\,\left( \tfrac{1}{\Omega }\left( \text{Tr}\left( \Delta \beta \sigma {{\left( \Delta \beta  \right)}^{T}} \right) \right) \right) \\ 
 & =\arg \underset{S}{\mathop{\min }}\,\left( \tfrac{1}{\Omega }\left( \text{Tr}\left( {{S}^{T}}\left( \Delta \alpha \sigma {{\left( \Delta \alpha  \right)}^{T}} \right)S \right) \right) \right),  
\end{split}
\end{equation} 
where $\sigma $ is as
\begin{equation} \label{ZEqnNum475164} 
\sigma =I+c\left(\eta -W\right),                                                          
\end{equation} 
and
\begin{equation} \label{24)} 
\Omega ={\rm Tr}\left(S^{T} \left(\Delta \alpha \eta \left(\Delta \alpha \right)^{T} \right)S\right).                                                   
\end{equation} 
At a particular $\Delta \alpha $ \eqref{ZEqnNum837426} and $\sigma $ \eqref{ZEqnNum475164}, the $S$ stabilizer matrix in equation \eqref{ZEqnNum782750} is evaluated via
\begin{equation} \label{ZEqnNum998235} 
\left(\Delta \alpha \sigma \left(\Delta \alpha \right)^{T} \right)S=\lambda \left(\Delta \alpha \eta \left(\Delta \alpha \right)^{T} \right)S, 
\end{equation} 
where $\lambda $ is a diagonal matrix of eigenvalues \cite{ref23,ref24}.

Algorithm A.1 (${\rm {\mathcal A}}_{S} $) gives the method for stabilizing the optimal state of the quantum computer.

 \setcounter{algocf}{0}
\begin{algo}
  \DontPrintSemicolon
\caption{Stabilization of the Optimal State of the Quantum Computer}
\textbf{Step 1}. Set the $R$ number of sequences for the quantum state stabilization. Formulate $\alpha $ \eqref{ZEqnNum579558} via $R$ gate parameter vectors  $\vec{\theta }_{r}^{*} $, $r=1,\ldots ,R$. 

\textbf{Step 2}. Set $\kappa $, and determine the $\omega _{rs} $ weight coefficients via equation \eqref{ZEqnNum241414} for all $r$ and $s$. 

\textbf{Step 3}. Set $W$, $\eta $, and $\sigma $ \eqref{ZEqnNum475164}.

\textbf{Step 4}. Compute the $S$ stabilizer matrix via equation \eqref{ZEqnNum998235}.

\textbf{Step 5}. Output $\beta =S^{T} \alpha $ via equation \eqref{ZEqnNum481709} for the stabilization of the optimal quantum state ${| \vec{\theta }^{*} \rangle} $ via the stable state ${\left| \vec{\varphi } \right\rangle} $ \eqref{ZEqnNum600779} through $R$ sequences of the quantum computer.
\end{algo} 

\end{proof}

\section{Learning the Stable Quantum State and Stability Class}
\label{sec4}
\begin{lemma}
The stabilized sequences of the quantum computer can be determined via unsupervised learning.
\end{lemma}
\begin{proof}
Algorithm 1 with the objective function \eqref{ZEqnNum782750} can be used to formulate an unsupervised learning framework to find the stabilized unitaries. The steps are detailed in Procedure 1 (${\rm {\mathcal P}}_{S} $). 

 \setcounter{algocf}{0}
\begin{proced}
  \DontPrintSemicolon
\caption{Unsupervised Learning of Stable Quantum Evolutions}
\textbf{Step 1}. Construct a ${\rm {\mathcal T}}$ training set of random gate parameters of the $QG$-structure of the quantum computer, as
\begin{equation} \label{26)} 
{\rm {\mathcal T}}=\left(X_{1} ,\ldots ,X_{q} \right),                                                               
\end{equation} 
where $X_{i} $ is a $K$-dimensional random vector, $d\le R$, formulated as
\begin{equation} \label{27)} 
X_{i} =\left[\theta _{i,1} ,\ldots ,\theta _{i,K} \right],                                                                 
\end{equation} 
where $\theta _{i,j} $ is the gate parameter of $U_{j} $ in $X_{i} $, and $j$ is a random number. 

\textbf{Step 2}. Determine the $S$ stabilizer matrix via Algorithm 1.

\textbf{Step 3}. Compute $Z=\left[z_{1} ,\ldots ,z_{q} \right]$ as
\begin{equation} \label{28)} 
Z=\left(S^{T} {\rm {\mathcal T}}^{T} \right)^{T} ={\rm {\mathcal T}}S,                                                         
\end{equation} 
and set $B=\left[b_{1} ,\ldots ,b_{q} \right]$ as
\begin{equation} \label{29)} 
B=-Z^{T} \bar{{\rm {\mathcal T}}},                                                                   
\end{equation} 
where $\bar{{\rm {\mathcal T}}}$ is the mean of all training samples \cite{ref24}. 

\textbf{Step 4}. For a given $\vec{\theta }_{r}^{*} $ of an $r$-th sequence, learn output $Y_{r} $ as
\begin{equation} \label{30)} 
Y_{r} =Z^{T} \vec{\theta }_{r}^{*} +B=\left({\rm {\mathcal T}}S\right)^{T} \vec{\theta }_{r}^{*} -\left({\rm {\mathcal T}}S\right)^{T} \bar{{\rm {\mathcal T}}}. 
\end{equation} 

\textbf{Step 5}. For an $i$-th gate parameter $\theta _{r,i}^{*} $,  learn the $j$-th output $y_{i,j}^{\left(r\right)} $ as
\begin{equation} \label{31)} 
y_{i,j}^{\left(r\right)} =z_{j} \otimes \theta _{r,i}^{*} +b_{j} ,                                                       
\end{equation} 
from which a statistical average for a given $i$, $i=1,\ldots ,L$, is 
\begin{equation} \label{ZEqnNum935956} 
\tilde{y}_{i}^{\left(r\right)} ={\textstyle\frac{1}{q}} \sum _{j=1}^{q}|y_{i,j}^{(r)}| ,                                                         
\end{equation} 
with difference $\Delta \tilde{y}_{i}^{\left(r\right)} $ as
\begin{equation} \label{33)} 
\Delta \tilde{y}_{i}^{\left(r\right)} =|\tilde{y}_{i}^{(r)} -\tilde{y}_{i+1}^{(r)}|.                                                           
\end{equation} 

\textbf{Step 6}. Repeat step 5 for all $i$. 

\textbf{Step 7}. Repeat steps 1-5 for the $R$ sequences. 

\end{proced}

\end{proof}

\subsection{Learning the Sequence Stability of Stabilized Quantum States}
\begin{proposition}
The stability of a given sequence $\vec{\varphi }_{r} $ can be characterized via $K$ stability levels. The sequence $\vec{\varphi }_{r} $ can be classified into $K$ stability classes from set ${\rm {\mathcal C}}$, 
\begin{equation} \label{ZEqnNum924379} 
{\rm {\mathcal C}}=\left\{C_{1} ,\ldots ,C_{K} \right\},                                                              
\end{equation} 
where $C_{k} $, $k=1,\ldots ,K$, is the $k$-th stability class. 
\end{proposition}
\begin{theorem}
The $C\left(\vec{\varphi }_{r} \right)$, $r=1,\ldots ,R$ stability class of a $\vec{\varphi }_{r} $ stabilized sequence, $\vec{\varphi }_{r} =\varphi _{r,1} ,\ldots ,\varphi _{r,L} $, of the quantum computer can be learned via $\phi _{k} \left(\vec{\varphi }_{r} \right)=\left(\nu _{k} \left(\vec{\varphi }_{r} \right)\right)^{T} f_{k}^{{\rm {\mathcal C}}} \left(\vec{\varphi }_{r} \right)$ quantities in the high-dimensional Hilbert space ${\rm {\mathcal H}}$, where $\nu _{k} \left(\varphi _{r,i} \right)={\textstyle\frac{1}{\pi }} \left(\varphi _{r,i} \right),$ and $f_{k}^{{\rm {\mathcal C}}} \left(\varphi _{r,i} \right)\in \left[0,1\right]$ is a probability.
\end{theorem}
\begin{proof}
Since the gate parameters are stabilized, the gate parameters $\vec{\varphi }_{r} $ and $\vec{\varphi }_{r+1} $ of the $r$-th and $\left(r+1\right)$-th sequences must be correlated in the stable system state ${\left| \vec{\varphi } \right\rangle} $ \eqref{ZEqnNum600779} of the quantum computer. 

Let $\beta $ from equation \eqref{ZEqnNum728789} be the $R$ stabilized sequences, where $\vec{\varphi }_{r} $ is the stabilized gate parameter vector of an $r$-th sequence of the quantum computer, and let ${\rm {\mathcal S}}$ be the set of all sequences of gate parameters as
\begin{equation} \label{35)} 
{\rm {\mathcal S}}=\bigcup _{\vec{\varphi }_{r} \in \beta }\left\{\left. \varphi _{r,i} \right|\varphi _{r,i} \in \vec{\varphi }_{r} \right\} .                                                     
\end{equation} 
For an $k$-th stabilization class $C_{k} $,  a probabilistic classifier function $f_{k}^{{\rm {\mathcal C}}} $ \cite{ref23,ref28} can be defined as
\begin{equation} \label{ZEqnNum469865} 
f_{k}^{{\rm {\mathcal C}}} :{\rm {\mathcal S}}\to \left[0,1\right].                                                                
\end{equation} 
The goal is to learn a function that maps any $\vec{\varphi }_{r} $ sequence to the correct stability class. Applying equation \eqref{ZEqnNum469865} on a given sequence $\vec{\varphi }_{r} $, i.e., $f_{k}^{{\rm {\mathcal C}}} \left(\vec{\varphi }_{r} \right)$ therefore maps $\vec{\varphi }_{r} $ to a given stability class via the classification of each $L$ gate parameter of the sequence. 

Thus, an $i$-th stabilized gate parameter $\varphi _{r,i} $ of an $r$-th sequence $\vec{\varphi }_{r} $  can be also classified into a particular stabilization class from ${\rm {\mathcal C}}$ \eqref{ZEqnNum924379}. The $f_{k}^{{\rm {\mathcal C}}} \left(\varphi _{r,i} \right)\in \left[0,1\right]$, $k=1,\ldots ,K$, classifier \eqref{ZEqnNum469865} is trained to classify \cite{ref28} each of the $\varphi _{r,i} $ gate parameters of $\vec{\varphi }_{r} $, $i=1,\ldots ,L$ via outputting a corresponding probability that $\varphi _{r,i} $ belongs to a given $C_{k} $ class. For a particular $\varphi _{r,i} $, the sum of the probabilities yields
\begin{equation} \label{37)} 
\sum _{k=1}^{K}f_{k}^{{\rm {\mathcal C}}} \left(\varphi _{r,i} \right) =1, 
\end{equation} 
for all $i$.

Then, let $\nu _{k} \left(\varphi _{r,i} \right)\ge 0$ be a weight parameter associated with a particular $\varphi _{r,i} $ and $k$-th class $C_{k} $, defined as
\begin{equation} \label{38)} 
\nu _{k} \left(\varphi _{r,i} \right)={\textstyle\frac{1}{\pi }} \left(\varphi _{r,i} \right), 
\end{equation} 
which normalizes $\varphi _{r,i} $ into the range of $\left[0,1\right]$, $\nu _{k} \left(\varphi _{r,i} \right)\in \left[0,1\right]$. 

For an $r$-th sequence $\vec{\varphi }_{r} $, a $\nu _{k} \left(\vec{\varphi }_{r} \right)$ collection can be defined as 
\begin{equation} \label{ZEqnNum624850} 
\nu _{k} \left(\vec{\varphi }_{r} \right)=\left[\nu _{k} \left(\varphi _{r,1} \right),\ldots ,\nu _{k} \left(\varphi _{r,L} \right)\right],                                          
\end{equation} 
where $\sum _{i=1}^{L}\nu _{k} \left(\varphi _{r,i} \right) =1$.

From equations \eqref{ZEqnNum469865} and \eqref{ZEqnNum624850}, the $\phi _{k} \left(\vec{\varphi }_{r} \right)$ evolution of a particular sequence $\vec{\varphi }_{r} $ with respect to a $k$-th class $C_{k} $  is defined as 
\begin{equation} \label{ZEqnNum396957} 
\begin{split}
   {{\phi }_{k}}\left( {{{\vec{\varphi }}}_{r}} \right)&={{\left( {{\nu }_{k}}\left( {{{\vec{\varphi }}}_{r}} \right) \right)}^{T}}f_{k}^{\mathcal{C}}\left( {{{\vec{\varphi }}}_{r}} \right) \\ 
 & =\left[ {{\nu }_{k}}\left( {{\varphi }_{r,1}} \right)f_{k}^{\mathcal{C}}\left( {{\varphi }_{r,1}} \right),\ldots ,{{\nu }_{k}}\left( {{\varphi }_{r,L}} \right)f_{k}^{\mathcal{C}}\left( {{\varphi }_{r,L}} \right) \right].  
\end{split}
\end{equation} 
Since the $\phi _{k} \left(\vec{\varphi }_{r} \right)$ term \eqref{ZEqnNum396957} is a non-linear map, the problem of correlation analysis \cite{ref23,ref28} between the inner products of non-linear functions $\phi _{k} \left(\vec{\varphi }_{r} \right)$ and $\phi _{l} \left(\vec{\varphi }_{r} \right)$ can be reformulated via a kernel machine ${\rm {\mathcal K}}$ \cite{ref25,ref26,ref27} as ${\rm {\mathcal K}}\left(\phi _{k} \left(\vec{\varphi }_{r} \right),\phi _{l} \left(\vec{\varphi }_{r} \right)\right)$, which yields a distance in a high-dimensional Hilbert space ${\rm {\mathcal H}}$. This distance in ${\rm {\mathcal H}}$ can therefore be used as a metric to describe the correlation between $\phi _{k} \left(\vec{\varphi }_{r} \right)$ and $\phi _{l} \left(\vec{\varphi }_{r} \right)$.

Let ${\rm {\mathcal X}}$ be the input space and let ${\rm {\mathcal K}}$ be an arbitrary kernel machine, defined for a given $x,y\in {\rm {\mathcal X}}$ via the kernel function
\begin{equation} \label{ZEqnNum980152} 
{\rm {\mathcal K}}\left(x,y\right)=\Gamma \left(x\right)^{T} \Gamma \left(y\right),                                                           
\end{equation} 
where
\begin{equation} \label{ZEqnNum942922} 
\Gamma :{\rm {\mathcal X}}\to {\rm {\mathcal H}} 
\end{equation} 
is a nonlinear map from ${\rm {\mathcal X}}$ to the high-dimensional reproducing kernel Hilbert space (RKHS) ${\rm {\mathcal H}}$ associated with ${\rm {\mathcal K}}$. Without a loss of generality, $\dim \left({\rm {\mathcal H}}\right){\rm \gg }\dim \left({\rm {\mathcal X}}\right)$, and we assume that the map $\Gamma $ in equation \eqref{ZEqnNum942922} has no inverse. 

Then, for a $\phi _{k} \left(\vec{\varphi }_{r} \right)$ and $\phi _{l} \left(\vec{\varphi }_{r} \right)$, let $\rho \left(\phi _{k} \left(\vec{\varphi }_{r} \right),\phi _{l} \left(\vec{\varphi }_{r} \right)\right)\to {\rm {\mathcal H}}$ be the correlation identifier, as
\begin{equation} \label{ZEqnNum540542} 
\begin{split}
&   \rho \left( {{\phi }_{k}}\left( {{{\vec{\varphi }}}_{r}} \right),{{\phi }_{l}}\left( {{{\vec{\varphi }}}_{r}} \right) \right)\\&=\mathcal{K}\left( {{\phi }_{k}}\left( {{{\vec{\varphi }}}_{r}} \right),{{\phi }_{l}}\left( {{{\vec{\varphi }}}_{r}} \right) \right) \\ 
 & =\sum\limits_{i=1}^{L}{\mathcal{K}\left( {{\nu }_{k}}\left( {{\varphi }_{r,i}} \right)f_{k}^{\mathcal{C}}\left( {{\varphi }_{r,i}} \right),{{\nu }_{l}}\left( {{\varphi }_{r,i}} \right)f_{l}^{\mathcal{C}}\left( {{\varphi }_{r,i}} \right) \right)}.  
\end{split}
\end{equation} 
Assuming that ${\rm {\mathcal K}}$ is a Gaussian kernel \cite{ref25,ref26,ref27} in equation \eqref{ZEqnNum540542}, for an $i$-th gate parameter the kernel function is 
\begin{equation} \label{ZEqnNum528545} 
\begin{split}
  & \mathcal{K}\left( {{\nu }_{k}}\left( {{\varphi }_{r,i}} \right)f_{k}^{\mathcal{C}}\left( {{\varphi }_{r,i}} \right),{{\nu }_{l}}\left( {{\varphi }_{r,i}} \right)f_{l}^{\mathcal{C}}\left( {{\varphi }_{r,i}} \right) \right) \\ 
 & =\exp \left( -\tfrac{1}{c}{{f}_{d}}\left( {{\nu }_{k}}\left( {{\varphi }_{r,i}} \right)f_{k}^{\mathcal{C}}\left( {{\varphi }_{r,i}} \right),{{\nu }_{l}}\left( {{\varphi }_{r,i}} \right)f_{l}^{\mathcal{C}}\left( {{\varphi }_{r,i}} \right) \right) \right),  
\end{split}
\end{equation} 
where $c=2\sigma ^{2} $, while $f_{d} \left(\cdot \right)$ yields the ${\rm L}2$-distance in ${\rm {\mathcal H}}$,
\begin{equation} \label{45)} 
\begin{split}
&f_{d} \left(\nu _{k} \left(\varphi _{r,i} \right)f_{k}^{{\rm {\mathcal C}}} \left(\varphi _{r,i} \right),\nu _{l} \left(\varphi _{r,i} \right)f_{l}^{{\rm {\mathcal C}}} \left(\varphi _{r,i} \right)\right)\\=&\left\| \nu _{k} \left(\varphi _{r,i} \right)f_{k}^{{\rm {\mathcal C}}} \left(\varphi _{r,i} \right)-\nu _{l} \left(\varphi _{r,i} \right)f_{l}^{{\rm {\mathcal C}}} \left(\varphi _{r,i} \right)\right\| _{2}^{2} .   
\end{split}
\end{equation} 
For a given $\phi _{k} \left(\vec{\varphi }_{r} \right)$ and $\phi _{l} \left(\vec{\varphi }_{r} \right)$, an $f_{A} \left(\phi _{c} \left(\vec{\varphi }_{r} \right),\phi _{c} \left(\vec{\varphi }_{r} \right)\right)\to \delta _{+}^{K} $ average is yielded as
\begin{equation} \label{46)} 
\begin{split}
& \varsigma \left( {{\phi }_{c}}\left( {{{\vec{\varphi }}}_{r}} \right),{{\phi }_{c}}\left( {{{\vec{\varphi }}}_{r}} \right) \right)={{\left( {{\phi }_{c}}\left( {{{\vec{\varphi }}}_{r}} \right) \right)}^{T}}{{\phi }_{c}}\left( {{{\vec{\varphi }}}_{r}} \right) \\ 
  & =\sum\limits_{i=1}^{L}{\nu _{c}^{2}\left( {{\varphi }_{r,i}} \right){{\left( f_{c}^{\mathcal{C}}\left( {{\varphi }_{r,i}} \right) \right)}^{2}}}\le \sum\limits_{i=1}^{L}{{{\nu }_{c}}\left( {{\varphi }_{r,i}} \right)f_{c}^{\mathcal{C}}\left( {{\varphi }_{r,i}} \right)},  
\end{split}
\end{equation} 
where $\delta _{+}^{K} $ refers to the space of $K\times K$ symmetric positive semi-definite matrices \cite{ref26,ref27,ref28}, while the $\iota $ inner products of $\phi _{k} \left(\vec{\varphi }_{r} \right)$ and $\phi _{l} \left(\vec{\varphi }_{r} \right)$ are represented in $\delta _{+}^{K} $ via $\iota \left(\phi _{k} \left(\vec{\varphi }_{r} \right),\phi _{l} \left(\vec{\varphi }_{r} \right)\right)\to \delta _{+}^{K} $ as 
\begin{equation} \label{47)} 
\begin{split}
 &  \iota \left( {{\phi }_{k}}\left( {{{\vec{\varphi }}}_{r}} \right),{{\phi }_{l}}\left( {{{\vec{\varphi }}}_{r}} \right) \right)={{\left( {{\phi }_{k}}\left( {{{\vec{\varphi }}}_{r}} \right) \right)}^{T}}{{\phi }_{l}}\left( {{{\vec{\varphi }}}_{r}} \right) \\ 
 & =\sum\limits_{i=1}^{L}{{{\nu }_{k}}\left( {{\varphi }_{r,i}} \right){{\nu }_{l}}\left( {{\varphi }_{r,i}} \right)\left( f_{k}^{\mathcal{C}}\left( {{\varphi }_{r,i}} \right) \right)}\left( f_{l}^{\mathcal{C}}\left( {{\varphi }_{r,i}} \right) \right).  
\end{split}
\end{equation} 
The $\vec{\varphi }_{r} $ sequence is classified into a given class from set ${\rm {\mathcal C}}$, as given in Algorithm 2 (${\rm {\mathcal A}}_{C} $).

 \setcounter{algocf}{1}
\begin{algo}
  \DontPrintSemicolon
\caption{Learning the Classification of the Stabilized Quantum States of the Quantum Computer}

\textbf{Step 1}. Let $\vec{\varphi }_{r} $ be the $r$-th sequence of the quantum computer, with the $L$ stabilized  gate parameters $\varphi _{r,1} ,\ldots ,\varphi _{r,L} $.

\textbf{Step 2}. Define set ${\rm {\mathcal C}}$ of the $K$ stability classes via \eqref{ZEqnNum924379}.

\textbf{Step 3}. Select $k$ that identifies $k$-th stability class $C_{k} $, and learn function $\rho \left(\phi _{k} \left(\vec{\varphi }_{r} \right),\phi _{l} \left(\vec{\varphi }_{r} \right)\right)$ \eqref{ZEqnNum540542} using the ${\rm {\mathcal K}}$ kernel machine \eqref{ZEqnNum980152} for all $l$, $l\ne k$.  

\textbf{Step 4}. Determine $\ell _{k} \left(\vec{\varphi }_{r} \right)=\mathop{\max }\limits_{\forall l} \rho \left(\phi _{k} \left(\vec{\varphi }_{r} \right),\phi _{l} \left(\vec{\varphi }_{r} \right)\right)$.

\textbf{Step 5}. Repeat steps 3-4 for all $k$.

\textbf{Step 6}. Determine $\xi \left(\vec{\varphi }_{r} \right)=\mathop{\max }\limits_{\forall k} \phi _{k} \left(\vec{\varphi }_{r} \right)$. 

\textbf{Step 7}. Classify $\vec{\varphi }_{r} $ into stability class $C\left(\vec{\varphi }_{r} \right)$ via the set ${\rm {\mathcal C}}$ as
\[C\left(\vec{\varphi }_{r} \right)=\xi \left(\vec{\varphi }_{r} \right)C_{p} +\ell _{k} \left(\vec{\varphi }_{r} \right)C_{q} ,\] 
where $p$ indexes the maximal $\phi _{k} \left(\vec{\varphi }_{r} \right)$ in $\xi \left(\vec{\varphi }_{r} \right)$, while $q$ indexes the maximal $\phi _{l} \left(\vec{\varphi }_{r} \right)$ in $\xi \left(\vec{\varphi }_{r} \right)$.

\textbf{Step 8}. Repeat steps 1-8 for all $r$.

\textbf{Step 9}. Output the stability classes $C\left(\beta \right)=\left[C\left(\vec{\varphi }_{1} \right),\ldots ,C\left(\vec{\varphi }_{R} \right)\right]^{T} $ of the stabilized quantum states ${\left| \vec{\varphi }_{r}  \right\rangle} $, $r=1,\ldots ,R$ of the quantum computer. 

\end{algo} 

\end{proof}

\section{Numerical Evaluation}
\label{nume}
\subsection{System Stability}
Let ${\left| \phi  \right\rangle} $ be the stabilized system state of the quantum computer formulated by $R$ output systems, ${\left| \vec{\varphi }_{r}  \right\rangle} $, $r=1,\ldots ,R$, as 
\begin{equation} \label{ZEqnNum703387} 
{\left| \phi  \right\rangle} ={\left| \vec{\varphi }_{1}  \right\rangle} \otimes \cdots \otimes {\left| \vec{\varphi }_{R}  \right\rangle} ,   
\end{equation} 
with gate parameters $\beta $, as given in \eqref{ZEqnNum728789}.

Then, let ${\left| \phi ^{*}  \right\rangle} $ be a target stabilized system of the quantum computer, as
\begin{equation} \label{2)} 
{\left| \phi ^{*}  \right\rangle} ={\left| \vec{\varphi }_{1}^{*}  \right\rangle} \otimes \cdots \otimes {\left| \vec{\varphi }_{R}^{*}  \right\rangle} ,   
\end{equation} 
with target gate parameters $\beta ^{*} $, as
\begin{equation} \label{3)} 
\beta ^{*} =\left[\vec{\varphi }_{1}^{*} ,\ldots ,\vec{\varphi }_{R}^{*} \right] 
\end{equation} 
where $\vec{\varphi }_{r}^{*} =\left[\varphi _{r,1}^{*} ,\ldots ,\varphi _{r,L}^{*} \right]^{T} $. 

Then, let $\left[\beta \right]_{rl}$ refer to the gate parameter $\varphi _{r,l} $ of an $l$-th unitary of an $r$-th running sequence of the quantum computer, $l=1,\ldots ,L$, $r=1,\ldots ,R$, in the state ${\left| \phi  \right\rangle} $, and let $\left[\beta ^{*} \right]_{rl} $ identify the target gate parameter $\varphi _{r,l}^{*} $ in state ${\left| \phi ^{*}  \right\rangle} $.

Then, let $D\left(\left. \beta \right\| \beta ^{*} \right)$ be the relative entropy between $\beta $ and $\beta ^{*} $, as
\begin{equation} \label{ZEqnNum686826} 
D\left(\left. \beta \right\| \beta ^{*} \right)=\sum _{r,l}\left(\left[\beta \right]_{rl} \log {\textstyle\frac{\left[\beta \right]_{rl} }{\left[\beta ^{*} \right]_{rl} }} +\left[\beta ^{*} \right]_{rl} -\left[\beta \right]_{rl} \right).  
\end{equation} 
where $D\left(\left. \beta \right\| \beta ^{*} \right)\ge 0$, and let $f_{D\left(\left. \beta \right\| \beta ^{*} \right)} \left(r\right)\ge 0$ be a function that returns the value of the relative entropy function for an $r$-th running sequence as
\begin{equation} \label{ZEqnNum729266} 
f_{D\left(\left. \beta \right\| \beta ^{*} \right)} \left(r\right)=\sum _{l}\left(\left[\beta \right]_{rl} \log {\textstyle\frac{\left[\beta \right]_{rl} }{\left[\beta ^{*} \right]_{rl} }} +\left[\beta ^{*} \right]_{rl} -\left[\beta \right]_{rl} \right).  
\end{equation} 
Let $f_{D\left(\left. \beta \right\| \beta ^{*} \right)}^{*} \left(r\right)$ be a target value for function \eqref{ZEqnNum729266}, and let $\Delta \left(f_{D\left(\left. \beta \right\| \beta ^{*} \right)} \left(r\right)\right)$  be the difference \cite{slow} between $f_{D\left(\left. \beta \right\| \beta ^{*} \right)}^{*} \left(r\right)$ and \eqref{ZEqnNum729266}, as
\begin{equation} \label{ZEqnNum982950} 
\Delta \left( {{f}_{D\left( \left. \beta  \right\|{{\beta }^{*}} \right)}}\left( r \right) \right)=\tfrac{1}{R}\int\limits_{1}^{R}{\partial _{D\left( \left. \beta  \right\|{{\beta }^{*}} \right)}^{2}\left( r \right)dr}, 
\end{equation} 
where ${{\partial }_{D\left( \left. \beta  \right\|{{\beta }^{*}} \right)}}\left( r \right)$ is the derivative of $f_{D\left(\left. \beta \right\| \beta ^{*} \right)} \left(r\right)$. 

Using \eqref{ZEqnNum982950}, we define a stability parameter $\delta $ to quantify the variation of the ${\left| \vec{\varphi }_{r}  \right\rangle} $ stabilized system state of the $r$-th running sequence of the quantum computer, as
\begin{equation} \label{ZEqnNum826820} 
\delta \left(r\right)=\left({\textstyle\frac{R}{2\pi }} \sqrt{\Delta \left(f_{D\left(\left. \beta \right\| \beta ^{*} \right)} \left(r\right)\right)} \right)^{-1} .  
\end{equation} 
For analytical purposes, let us assume that $f_{D\left(\left. \beta \right\| \beta ^{*} \right)} \left(r\right)$ oscillates between a minimal value $\gamma \ge 0$, and a maximal value $\gamma \le \lambda \le 1$, defined as
\begin{equation} \label{8)} 
\gamma =\arg \mathop{\min }\limits_{\forall r} \left(f_{D\left(\left. \beta \right\| \beta ^{*} \right)} \left(r\right)\right),   
\end{equation} 
and
\begin{equation} \label{9)} 
\lambda =\arg \mathop{\max }\limits_{\forall r} \left(f_{D\left(\left. \beta \right\| \beta ^{*} \right)} \left(r\right)\right),   
\end{equation} 
therefore $f_{D\left(\left. \beta \right\| \beta ^{*} \right)} \left(r\right)$ can be rewritten as
\begin{equation} \label{ZEqnNum745453} 
f_{D\left(\left. \beta \right\| \beta ^{*} \right)} \left(r\right)=c\sin \left(N2\pi {\textstyle\frac{r}{R}} \right)+{{\mathbb{E}}}\left(D\left(\left. \beta \right\| \beta ^{*} \right)\right),   
\end{equation} 
where $c$ is a constant, set as
\begin{equation} \label{11)} 
c={\textstyle\frac{1}{2}} \left(\lambda -\gamma \right),   
\end{equation} 
while $0\le {{\mathbb{E}}}\left(D\left(\left. \beta \right\| \beta ^{*} \right)\right)\le 1$ is an expected value of \eqref{ZEqnNum686826}, set as
\begin{equation} \label{12)} 
{{\mathbb{E}}}\left(D\left(\left. \beta \right\| \beta ^{*} \right)\right)=c+\gamma ,   
\end{equation} 
while $N$ is the number of oscillations. 

Therefore, \eqref{ZEqnNum982950} can be evaluated as
\begin{equation} \label{ZEqnNum728565} 
\begin{split}
   \Delta \left( {{f}_{D\left( \left. \beta  \right\|{{\beta }^{*}} \right)}}\left( r \right) \right)&=\tfrac{1}{R}\int\limits_{1}^{R}{\tfrac{2{{N}^{2}}4{{\pi }^{2}}}{{{R}^{2}}}{{\cos }^{2}}\left( N2\pi \tfrac{r}{R} \right)dr} \\ 
 & =\tfrac{2{{N}^{2}}4{{\pi }^{2}}}{{{R}^{2}}}\tfrac{1}{N2\pi }\int\limits_{1}^{N2\pi }{2{{\cos }^{2}}\left( {{r}'} \right)d{r}'} \\ 
 & =\tfrac{{{N}^{2}}4{{\pi }^{2}}}{{{R}^{2}}},  
\end{split}
\end{equation} 
where $r\in \left[r_{0} ,r_{0} +R\right]$, with $r_{0} =1$.

Then, by using \eqref{ZEqnNum728565}, the quantity in \eqref{ZEqnNum826820} is as
\begin{equation} \label{ZEqnNum257554} 
\delta \left(r\right)=\left({\textstyle\frac{R}{2\pi }} \sqrt{{\textstyle\frac{N^{2} 4\pi ^{2} }{R^{2} }} } \right)^{-1} ={\textstyle\frac{1}{N}} ,   
\end{equation} 
that identifies the inverse of the number of oscillations. 

Therefore, \eqref{ZEqnNum257554} identifies the stability of the system state ${\left| \vec{\varphi }_{r}  \right\rangle} $ of the quantum computer in the $r$-th running sequence if $f_{D\left(\left. \beta \right\| \beta ^{*} \right)} \left(r\right)$ has the form of \eqref{ZEqnNum745453}. For an arbitrary $f_{D\left(\left. \beta \right\| \beta ^{*} \right)} \left(r\right)$, the stability parameter $\delta \left(r\right)$ is evaluated via \eqref{ZEqnNum826820}. The high value of $\delta \left(r\right)$ indicates that the stabilized system ${\left| \vec{\varphi }_{r}  \right\rangle} $ in \eqref{ZEqnNum703387} changes slowly. Particularly, if $\delta \left(r\right)\ge \delta ^{*} \left(r\right)$, where $\delta ^{*} \left(r\right)$ is a target value for $\delta \left(r\right)$, then the system state ${\left| \vec{\varphi }_{r}  \right\rangle} $ of the quantum computer is considered as stable.

The values of $f_{D\left(\left. \beta \right\| \beta ^{*} \right)} \left(r\right)$ \eqref{ZEqnNum745453} and $\delta \left(r\right)$ \eqref{ZEqnNum257554} for $R$ running sequences are depicted in \fref{figA1}.

\begin{center}
\begin{figure*}[!htbp]
\begin{center}
\includegraphics[angle = 0,width=1\linewidth]{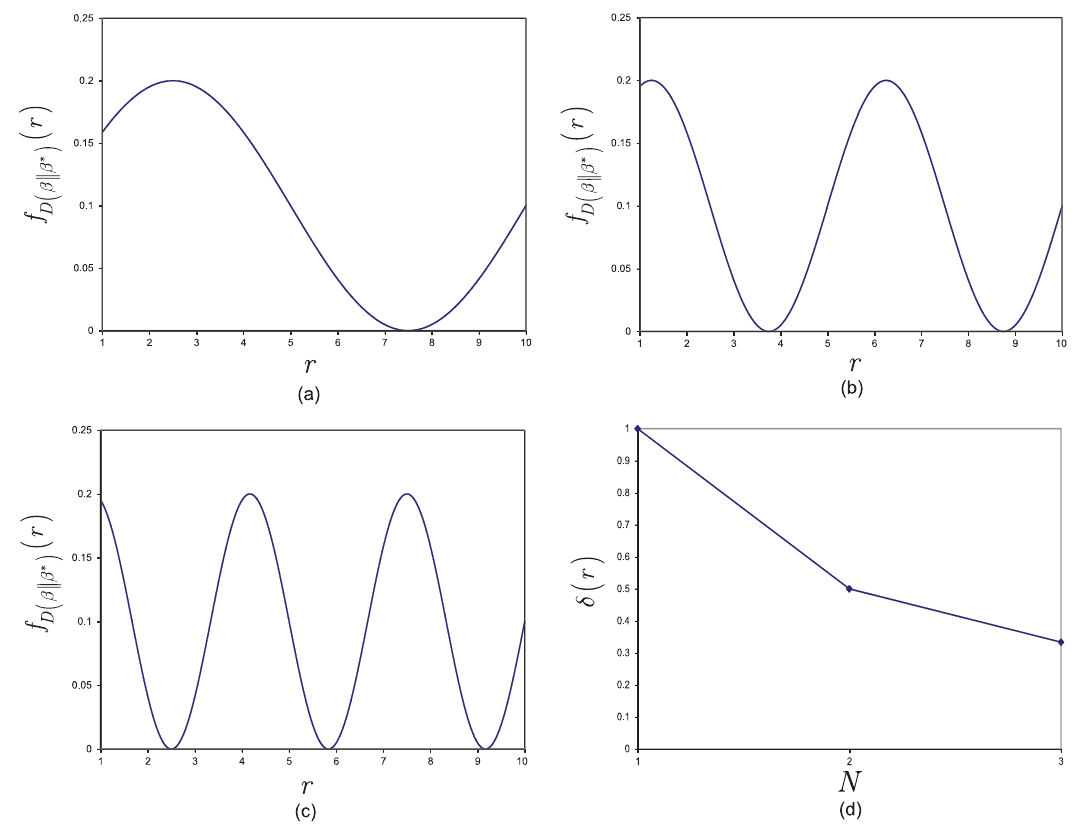}
\caption{The $f_{D\left(\left. \beta \right\| \beta ^{*} \right)} \left(r\right)$ relative entropy values between the gate parameters of the  stabilized system ${\left| \vec{\varphi }_{r}  \right\rangle} $ and target system ${\left| \vec{\varphi }_{r}^{*}  \right\rangle} $ for $R$ running sequences, $r=1,\ldots R$, $R=10$, ${{\mathbb{E}}}\left(D\left(\left. \beta \right\| \beta ^{*} \right)\right)=0.1$. (a) $N=1$. (b) $N=2$. (c) $N=3$. (d). Stability parameter $\delta \left(r\right)$ for the different relative entropy values.} 
 \label{figA1}
 \end{center}
\end{figure*}
\end{center}

\subsection{Gate Parameter Correlations}
Let ${\left| \vec{\varphi }_{r}  \right\rangle} $ be the stabilized state of the quantum computer in the $r$-th running sequence, with $\vec{\varphi }_{r} =\left[\varphi _{r,1} ,\ldots ,\varphi _{r,L} \right]^{T} $, and let ${\left| \vec{\varphi }_{r}^{*}  \right\rangle} $ be the target stabilized system state in the $r$-th running sequence, with $\vec{\varphi }_{r}^{*} =\left[\varphi _{r,1}^{*} ,\ldots ,\varphi _{r,L}^{*} \right]^{T} $. 

Then, let $\mu $ be a correlation coefficient \cite{slow} that measures the correlation of the gate parameters $\beta $ and $\beta ^{*} $ of ${\left| \phi  \right\rangle} $ \eqref{ZEqnNum703387} and ${\left| \phi ^{*}  \right\rangle} $ \eqref{ZEqnNum728789}, defined as
\begin{equation} \label{ZEqnNum865688} 
\mu \left(\beta ,\beta ^{*} \right)=\left|{\textstyle\frac{F\left(\left(f\left(\vec{\varphi }_{r} \right)-F\left(f\left(\vec{\varphi }_{r} \right)\right)\right)\left(f\left(\vec{\varphi }_{r}^{*} \right)-F\left(f\left(\vec{\varphi }_{r}^{*} \right)\right)\right)\right)}{\sqrt{F\left(\left(f\left(\vec{\varphi }_{r} \right)-F\left(f\left(\vec{\varphi }_{r} \right)\right)\right)^{2} \right)F\left(\left(f\left(\vec{\varphi }_{r}^{*} \right)-F\left(f\left(\vec{\varphi }_{r}^{*} \right)\right)\right)^{2} \right)} }} \right|,  
\end{equation} 
where $\left|\cdot \right|$ is the absolute value, $f\left(\vec{\varphi }_{r} \right)$ is a function of $r$ that represents the values of the gate parameter vector $\vec{\varphi }_{r} $, while $F\left(\cdot \right)$ is defined over $r\in \left[1,R\right]$, as
\begin{equation} \label{ZEqnNum418031} 
F\left( f\left( x \right) \right)=\tfrac{1}{R}\int\limits_{1}^{R}{f\left( x \right)dr}.
\end{equation} 
For illustration purposes, let us assume that $L=1$, and $f\left(\vec{\varphi }_{r} \right)$ is as
\begin{equation} \label{ZEqnNum843455} 
f\left(\vec{\varphi }_{r} \right)=X\cos ^{2} \left(CN2\pi {\textstyle\frac{r}{R}} \right),  
\end{equation} 
where we set $X$ as $X={2C^{2} N^{2} 4\pi ^{2} \mathord{\left/ {\vphantom {2C^{2} N^{2} 4\pi ^{2}  R^{2} }} \right. \kern-\nulldelimiterspace} R^{2} } $, while $C>0$ is a constant, thus \eqref{ZEqnNum418031} is evaluated as
\begin{equation} \label{18)} 
\begin{split}
   F\left( f\left( {{{\vec{\varphi }}}_{r}} \right) \right)&=\tfrac{1}{R}\int\limits_{1}^{R}{\tfrac{2{{C}^{2}}{{N}^{2}}4{{\pi }^{2}}}{{{R}^{2}}}{{\cos }^{2}}\left( CN2\pi \tfrac{r}{R} \right)dr} \\ 
 & =\tfrac{2{{C}^{2}}{{N}^{2}}4{{\pi }^{2}}}{{{R}^{2}}}\tfrac{1}{N2\pi }\int\limits_{1}^{N2\pi }{2{{\cos }^{2}}\left( {{r}'} \right)d{r}'} \\ 
 & =\tfrac{{{C}^{2}}{{N}^{2}}4{{\pi }^{2}}}{{{R}^{2}}}.  
\end{split}
\end{equation} 
For the target system ${\left| \vec{\varphi }_{r}^{*}  \right\rangle} $, the constant $C^{*} $ is set as
\begin{equation} \label{ZEqnNum876561} 
f\left(\vec{\varphi }_{r}^{*} \right)={\textstyle\frac{2\left(C^{*} \right)^{2} N^{2} 4\pi ^{2} }{R^{2} }} \cos ^{2} \left(C^{*} N2\pi {\textstyle\frac{r}{R}} \right),  
\end{equation} 
while for ${\left| \vec{\varphi }_{r}  \right\rangle} $, we set $C$ as $C>C^{*} $, thus \eqref{ZEqnNum865688} can be evaluated as
\begin{equation} \label{ZEqnNum426561} 
\mu \left(\beta ,\beta ^{*} \right)=\left|{\textstyle\frac{F\left(\left(f\left(\vec{\varphi }_{r} \right)-{\textstyle\frac{C^{2} N^{2} 4\pi ^{2} }{R^{2} }} \right)\left(f\left(\vec{\varphi }_{r}^{*} \right)-{\textstyle\frac{\left(C^{*} \right)^{2} N^{2} 4\pi ^{2} }{R^{2} }} \right)\right)}{\sqrt{F\left(\left(f\left(\vec{\varphi }_{r} \right)-{\textstyle\frac{C^{2} N^{2} 4\pi ^{2} }{R^{2} }} \right)^{2} \right)F\left(\left(f\left(\vec{\varphi }_{r}^{*} \right)-{\textstyle\frac{\left(C^{*} \right)^{2} N^{2} 4\pi ^{2} }{R^{2} }} \right)^{2} \right)} }} \right|,   
\end{equation} 
where
\begin{equation} 
\begin{split}
  & F\left( \left( f\left( {{{\vec{\varphi }}}_{r}} \right)-\tfrac{{{C}^{2}}{{N}^{2}}4{{\pi }^{2}}}{{{R}^{2}}} \right)\left( f\left( \vec{\varphi }_{r}^{*} \right)-\tfrac{{{\left( {{C}^{*}} \right)}^{2}}{{N}^{2}}4{{\pi }^{2}}}{{{R}^{2}}} \right) \right) \\ 
 & =F\left( f\left( {{{\vec{\varphi }}}_{r}} \right)f\left( \vec{\varphi }_{r}^{*} \right)-f\left( {{{\vec{\varphi }}}_{r}} \right)\tfrac{{{\left( {{C}^{*}} \right)}^{2}}{{N}^{2}}4{{\pi }^{2}}}{{{R}^{2}}}-\tfrac{{{C}^{2}}{{N}^{2}}4{{\pi }^{2}}}{{{R}^{2}}}f\left( \vec{\varphi }_{r}^{*} \right)+\tfrac{{{C}^{2}}{{N}^{2}}4{{\pi }^{2}}}{{{R}^{2}}}\tfrac{{{\left( {{C}^{*}} \right)}^{2}}{{N}^{2}}4{{\pi }^{2}}}{{{R}^{2}}} \right) \\ 
 & =\tfrac{1}{R}\left( \tfrac{2{{\pi }^{3}}{{C}^{2}}{{\left( {{C}^{*}} \right)}^{2}}{{N}^{3}}\left( \left( {{C}^{*}}-C \right)\sin \left( N4\pi \left( {{C}^{*}}+C \right) \right)+\left( {{C}^{*}}+C \right)\sin \left( N4\pi \left( {{C}^{*}}-C \right) \right) \right)}{\left( {{\left( {{C}^{*}} \right)}^{2}}-{{C}^{2}} \right){{R}^{3}}} \right),  
\end{split}
\end{equation} 
and
\begin{equation} \label{21)} 
\begin{split}
   F\left( {{\left( f\left( {{{\vec{\varphi }}}_{r}} \right)-\tfrac{{{C}^{2}}{{N}^{2}}4{{\pi }^{2}}}{{{R}^{2}}} \right)}^{2}} \right)&=F\left( f{{\left( {{{\vec{\varphi }}}_{r}} \right)}^{2}}-2f\left( {{{\vec{\varphi }}}_{r}} \right)\tfrac{{{C}^{2}}{{N}^{2}}4{{\pi }^{2}}}{{{R}^{2}}}+{{\left( \tfrac{{{C}^{2}}{{N}^{2}}4{{\pi }^{2}}}{{{R}^{2}}} \right)}^{2}} \right) \\ 
 & =\tfrac{1}{R}\left( \tfrac{{{C}^{4}}{{N}^{4}}8{{\pi }^{4}}}{{{R}^{3}}} \right) \\ 
 & =\tfrac{{{C}^{4}}{{N}^{4}}8{{\pi }^{4}}}{{{R}^{4}}},  
\end{split}
\end{equation} 
thus \eqref{ZEqnNum426561} is simplified as
\begin{equation} \label{ZEqnNum407594} 
\mu \left(\beta ,\beta ^{*} \right)=\left|{\textstyle\frac{2\pi ^{3} C^{2} \left(C^{*} \right)^{2} N^{3} \left(\left(C^{*} -C\right)\sin \left(N4\pi \left(C^{*} +C\right)\right)+\left(C^{*} +C\right)\sin \left(N4\pi \left(C^{*} -C\right)\right)\right)}{\left(\left(\left(C^{*} \right)^{2} -C^{2} \right)R^{4} \right)\sqrt{{\textstyle\frac{C^{4} N^{4} 8\pi ^{4} }{R^{4} }} {\textstyle\frac{\left(C^{*} \right)^{4} N^{4} 8\pi ^{4} }{R^{4} }} } }} \right|.  
\end{equation}

The values of \eqref{ZEqnNum407594} are depicted in \fref{figA2}.

\begin{center}
\begin{figure*}[!htbp]
\begin{center}
\includegraphics[angle = 0,width=0.8\linewidth]{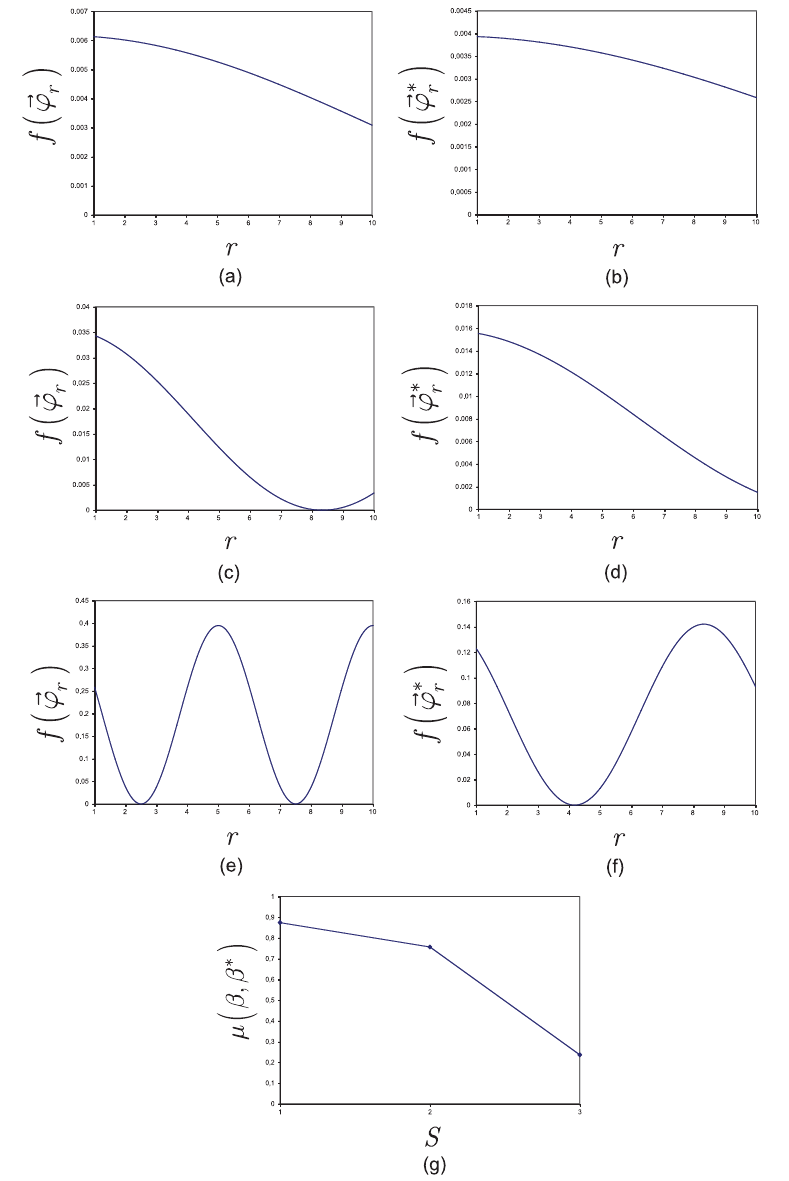}
\caption{Gate parameter values of $f\left(\vec{\varphi }_{r} \right)$ and $f\left(\vec{\varphi }_{r}^{*} \right)$ in function of running sequence $r$, $r=1,\ldots ,R$, $R=10$, $L=1$, for different  $N$, $C$ and $C^{*} $, $C\in \left[0,1\right]$, $C^{*} \in \left[0,1\right]$. (a-b) $N=1$, $C=0.125$, $C^{*} =0.1$. (c-d) $N=1$, $C=0.3$, $C^{*} =0.2$ (e-f) $N=2$, $C=0.5$, $C^{*} =0.3$. (g) The $\mu \left(\beta ,\beta ^{*} \right)$ correlation values between the gate parameters of (a-f), $S$ is an indexing parameter.} 
 \label{figA2}
 \end{center}
\end{figure*}
\end{center}

In \fref{figA3} the distribution of $\mu \left(\beta ,\beta ^{*} \right)$ in function of $C$ and $C^{*} $ are depicted, $C\in \left[0,1\right]$, $C^{*} \in \left[0,1\right]$ for different values of $N$, $f\left(\vec{\varphi }_{r} \right)$ and $f\left(\vec{\varphi }_{r}^{*} \right)$ are evaluated as given in \eqref{ZEqnNum843455} and \eqref{ZEqnNum876561}, $L=1$, and $R=10$ .

\begin{center}
\begin{figure*}[!htbp]
\begin{center}
\includegraphics[angle = 0,width=1\linewidth]{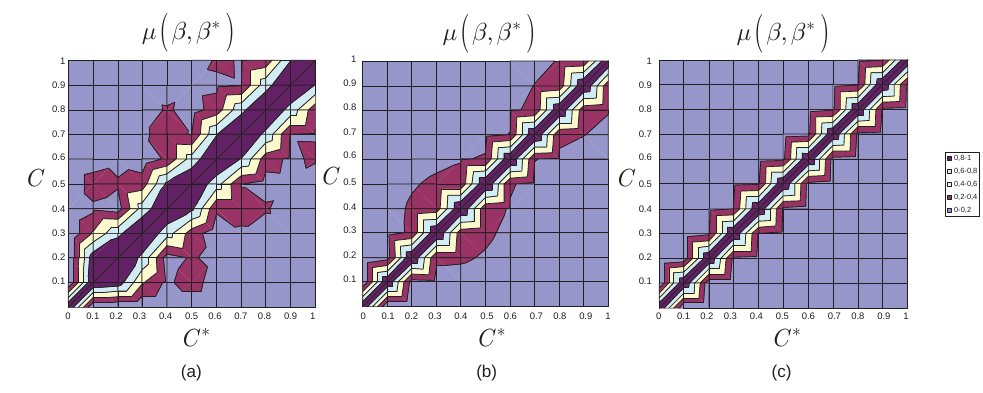}
\caption{The distribution of $\mu \left(\beta ,\beta ^{*} \right)$ in function of $C$ and $C^{*} $ at  $f\left(\vec{\varphi }_{r} \right)$ and $f\left(\vec{\varphi }_{r}^{*} \right)$, $L=1$, $R=10$. (a) $N=1$. (b) $N=2$. (c) $N=3$.} 
 \label{figA3}
 \end{center}
\end{figure*}
\end{center} 

\section{Conclusions}
\label{sec5}
Here, we defined a method for the learning of stable quantum evolutions in gate-model quantum computer architectures. The model stabilizes an optimal state of a quantum computer to maximize the particular objective function of an arbitrary problem fed into the quantum computer. The model learns a stabilizer matrix that stabilizes the state of the quantum computer through an arbitrary number of run sequences. We also defined a scheme to characterize the stability of the stabilized states via unsupervised learning of the stability classes of the stabilized sequences. The results are particularly useful for gate-based quantum computations and gate-model quantum computer architectures.

%

\section*{Acknowledgements}
The research reported in this paper has been supported by the National Research, Development and Innovation Fund (TUDFO/51757/2019-ITM, Thematic Excellence Program). This work was partially supported by the National Research Development and Innovation Office of Hungary (Project No. 2017-1.2.1-NKP-2017-00001), by the Hungarian Scientific Research Fund - OTKA K-112125 and in part by the BME Artificial Intelligence FIKP grant of EMMI (BME FIKP-MI/SC).

\newpage
\onecolumn
\appendix
\setcounter{table}{0}
\setcounter{figure}{0}
\setcounter{equation}{0}
\setcounter{algocf}{0}
\renewcommand{\thetable}{\Alph{section}.\arabic{table}}
\renewcommand{\thefigure}{\Alph{section}.\arabic{figure}}
\renewcommand{\theequation}{\Alph{section}.\arabic{equation}}
\renewcommand{\thealgocf}{\Alph{section}.\arabic{algocf}}

\section{Appendix}

\subsection{Abbreviations}
\begin{description}
\item[QG] Quantum Gate structure of a gate-model quantum computer
\item[RKHS] Reproducing Kernel Hilbert Space
\end{description}

\subsection{Notations}
\setlength{\arrayrulewidth}{0.1mm}
\setlength{\tabcolsep}{5pt}
\renewcommand{\arraystretch}{1.5}

The notations of the manuscript are summarized in  \tref{tab2}.
\begin{center}
\begin{longtable}{||l|p{4.5in}||}
\caption{Summary of notations.}
\label{tab2}
\endfirsthead
\endhead
\hline
$QG$ & Quantum gate structure of a gate-model quantum computer. \\ \hline 
$L$ & Number of unitaries in the $QG$ structure of the quantum computer. \\ \hline 
$U_{i} \left(\theta _{i} \right)$ & An $i$-th unitary gate, $U_{i} \left(\theta _{i} \right)=\exp \left(-i\theta _{i} P\right)$, where $P$ is a generalized Pauli operator formulated by a tensor product of Pauli operators $\left\{X,Y,Z\right\}$, while $\theta _{i} $ is referred to as the gate parameter associated to $U_{i} \left(\theta _{i} \right)$. \\ \hline 
${| \vec{\theta } \rangle} $ & System state of the quantum computer, ${| \vec{\theta } \rangle} =U_{L} \left(\theta _{L} \right)U_{L-1} \left(\theta _{L-1} \right)\ldots U_{1} \left(\theta _{1} \right)$, where $U_{i} \left(\theta _{i} \right)$ identifies an $i$-th unitary gate. \\ \hline 
$\vec{\theta }$ & Gate parameter vector, a collection of gate parameters of the $L$ unitaries, $\vec{\theta }=\left[\theta _{1} ,\ldots ,\theta _{L-1} ,\theta _{L} \right]^{T} $. \\ \hline 
$C$ & Classical objective function of a computational problem fed into the quantum computer.  \\ \hline 
$f(\vec{\theta })$ & Objective function of the quantum computer. \\ \hline 
${| \vec{\theta }^{*}  \rangle}$ & Optimal state of the quantum computer. \\ \hline 
$\vec{\theta }^{*} $ & Gate parameter vector in the ${| \vec{\theta }^{*}  \rangle}$ system state, $\vec{\theta }^{*} =\left[\theta _{1}^{*} ,\ldots ,\theta _{L}^{*} \right]^{T} $. \\ \hline 
$f(\vec{\theta }^{*})$ & Objective function value in the ${| \vec{\theta }^{*}  \rangle}$ system state. \\ \hline 
$P$ & Generalized Pauli operator formulated by the tensor product of Pauli operators $\left\{X,Y,Z\right\}$. \\ \hline 
${\left| \vec{\varphi } \right\rangle} $ & Stable system state with objective function $f(\vec{\varphi })=\langle \vec{\varphi }|C|\vec{\varphi }\rangle =f(\vec{\theta }^{*})$.  \\ \hline 
$\vec{\varphi }$ & Gate parameter vector associated to the stable system state ${\left| \vec{\varphi } \right\rangle} $, $\vec{\varphi }=\left[\varphi _{1} ,\ldots ,\varphi _{L} \right]^{T} $. \\ \hline 
$\vec{\theta }_{r}^{*} $ & Gate parameter vector, identifies the quantum state ${| \vec{\theta }_{r}^{*} \rangle}$ of an $r$-th running sequence, $r=1,\ldots ,R$, of the quantum computer,  $\vec{\theta }_{r}^{*} =[\theta _{r,1}^{*} ,\ldots ,\theta _{r,L}^{*} ]^{T}$. \\ \hline 
$\vec{\varphi }_{r} $ & Gate parameter vector, identifies the stabilized quantum state ${\left| \vec{\varphi }_{r}  \right\rangle} $ of an $r$-th sequence of the quantum computer, $\vec{\varphi }_{r} =\left[\varphi _{r,1} ,\ldots ,\varphi _{r,L} \right]^{T} $. \\ \hline 
$\alpha $ & Matrix, formulated via the $R$ sequences of the quantum computer, $\alpha =[\vec{\theta }_{1}^{*} ,\ldots ,\vec{\theta }_{R}^{*}]$. \\ \hline 
$\beta $ & Matrix, formulated via the $R$ stabilized sequences of the quantum computer, $\beta =\left[\vec{\varphi }_{1} ,\ldots ,\vec{\varphi }_{R} \right]$. \\ \hline 
$S$ & Stabilizer matrix, yields $\beta $ from  $\alpha $ as $\beta =S^{T} \alpha $, $S^{T} S=I$, where $I$ is the identity matrix. \\ \hline 
${\rm {\mathcal F}}$ & Solution framework. \\ \hline 
${\rm {\mathcal P}}_{S} $ & Stabilization procedure. \\ \hline 
${\rm {\mathcal A}}_{S} $ & Stabilization algorithm. \\ \hline 
${\rm {\mathcal A}}_{C} $ & Classification algorithm. \\ \hline 
$C\left(\beta \right)$ & Stability-class of $\beta $. \\ \hline 
$\Delta (\vec{\theta }_{r}^{*})$ & Vector, defined for an $r$-th sequence of the quantum computer, $\Delta (\vec{\theta }_{r}^{*})=\vec{\theta }_{r}^{*} -\vec{\theta }_{r+1}^{*} $. \\ \hline 
$\Delta \left(\vec{\varphi }_{r} \right)$ & Vector, defined for an $r$-th stabilized sequence of the quantum computer, $\Delta \left(\vec{\varphi }_{r} \right)=\vec{\varphi }_{r} -\vec{\varphi }_{r+1} $. \\ \hline 
$\Delta \alpha $ & A collection of $\Delta (\vec{\theta }_{r}^{*})$ vectors, $\Delta \alpha =[\Delta (\vec{\theta }_{1}^{*} ),\ldots ,\Delta (\vec{\theta }_{R-1}^{*})]$. \\ \hline 
$\Delta \beta $ & A collection of $\Delta \left(\vec{\varphi }_{r} \right)$ vectors, $\Delta \beta =\left[\Delta \left(\vec{\varphi }_{1} \right),\ldots ,\Delta \left(\vec{\varphi }_{R-1} \right)\right]$. \\ \hline 
$\chi $ & Sum defined via $\Delta \beta $ as $\chi =\sum _{r}^{R-1}\left\| \Delta \left(\vec{\varphi }_{r} \right)\right\| _{2}^{2}  $, where $\left\| \cdot \right\| _{2}^{2} $  is the squared ${\rm L}2$-norm. \\ \hline 
$\gamma _{rs} $ & Parameter, defined as $\gamma _{rs} =\omega _{rs} \left\| \Delta \left(\vec{\varphi }_{r} \right)-\Delta \left(\vec{\varphi }_{s} \right)\right\| _{2}^{2} $, where $\Delta \left(\vec{\varphi }_{r} \right)$ and $\Delta \left(\vec{\varphi }_{s} \right)$ are derived for an $r$-th and $s$-th sequences, $s>r$, while $\omega _{rs} $ is a weight coefficient. \\ \hline 
$\tau $ & A sum, defined for the $r=1,\ldots ,R-1$ sequences of the quantum computer, $\tau =\sum _{r}^{R-1}\sum _{s}^{R-1}\gamma _{rs}   .$ \\ \hline 
$F^{*} $ & Objective function of the stabilization procedure. \\ \hline 
$c$ & Regularization constant. \\ \hline 
$\omega _{rs} $ & Weight coefficient for the $r$-th and $s$-th sequences, $s>r$, \newline $\omega _{rs} =\left\{\begin{array}{l} {\exp \left(-{\textstyle\frac{\left\| \Delta (\vec{\theta }_{r}^{*})-\Delta \left(\vec{\theta }_{s}^{*} \right)\right\| ^{2} }{\zeta }} \right),{\rm if\; }\left(s-r\right)\le \kappa } \\ {0,{\rm \; otherwise}} \end{array}\right. ,$\newline where $\kappa $ and $\zeta $ are nonzero parameters. \\ \hline 
$W$ & Weight matrix, $W_{rs} =\omega _{rs} $. \\ \hline 
$\eta $ & Diagonal matrix of the weight coefficients, $\eta _{rr} =\sum _{s}\omega _{rs}  $, with relation $\left(\Delta \beta \right)^{T} \eta \Delta \beta =I$. \\ \hline 
$\sigma $ & Matrix, $\sigma =I+c\left(\eta -W\right)$. \\ \hline 
$\Omega $ & Parameter, $\Omega ={\rm Tr}(S^{T} (\Delta \alpha \eta (\Delta \alpha )^{T})S)$. \\ \hline 
$\lambda $ & Diagonal matrix of eigenvalues. \\ \hline 
${\rm {\mathcal T}}$ & Training set of random gate parameters of the $QG$-structure of the quantum computer, ${\rm {\mathcal T}}=\left(X_{1} ,\ldots ,X_{q} \right)$, where $X_{i} $ is a $d$-dimensional random vector. \\ \hline 
$\bar{{\rm {\mathcal T}}}$ & Mean of all training samples. \\ \hline 
$Y_{r} $ & Learned output for an $r$-th sequence. \\ \hline 
$y_{i,j}^{\left(r\right)} $ & Learned $j$-th output for an $i$-th unitary of an $r$-th sequence. \\ \hline 
$\Delta \tilde{y}_{i}^{\left(r\right)} $ & Difference, $\Delta \tilde{y}_{i}^{(r)} =|\tilde{y}_{i}^{(r)} -\tilde{y}_{i+1}^{(r)}|$. \\ \hline 
$C_{k} $ & A $k$-th stability class, $k=1,\ldots ,K$, for the classification of the stability of the stabilized sequences of $\beta $.  \\ \hline 
${\rm {\mathcal C}}$ & Set of $K$ stability classes, ${\rm {\mathcal C}}=\left\{C_{1} ,\ldots ,C_{K} \right\}$. \\ \hline 
$C\left(\vec{\varphi }_{r} \right)$ & Stability class of a stabilized sequence $\vec{\varphi }_{r} $. \\ \hline 
$C\left(\beta \right)$ & Stability classes of all $\vec{\varphi }_{r} $ stabilized sequences, $r=1,\ldots ,R$, of $\beta $, $C\left(\beta \right)=\left[C\left(\vec{\varphi }_{1} \right),\ldots ,C\left(\vec{\varphi }_{R} \right)\right]^{T} $. \\ \hline 
$\delta _{+}^{K} $ & Space of $K\times K$ symmetric positive semi-definite matrices. \\ \hline 
${\rm {\mathcal X}}$ & Input space. \\ \hline 
${\rm {\mathcal K}}$ & Kernel machine. \\ \hline 
${\rm {\mathcal H}}$ & Reproducing Kernel Hilbert Space (RKHS) associated with the kernel machine ${\rm {\mathcal K}}$. \\ \hline 
$\Gamma $ & A nonlinear map, $\Gamma :{\rm {\mathcal X}}\to {\rm {\mathcal H}}$, from ${\rm {\mathcal X}}$ to the high-dimensional Hilbert space ${\rm {\mathcal H}}$ associated with ${\rm {\mathcal K}}$. \\ \hline 
$f_{d} \left(x,y\right)$ & ${\rm L}2$ distance in ${\rm {\mathcal H}}$ , $f_{d} \left(x,y\right)=\left\| x-y\right\| _{2}^{2} $. \\ \hline 
$f_{k}^{{\rm {\mathcal C}}} $ & Probabilistic classifier function, $f_{k}^{{\rm {\mathcal C}}} :{\rm {\mathcal S}}\to \left[0,1\right]$, where ${\rm {\mathcal S}}=\bigcup _{\vec{\varphi }_{r} \in \beta }\left\{\left. \varphi _{r,i} \right|\varphi _{r,i} \in \vec{\varphi }_{r} \right\} $, $\sum _{k=1}^{K}f_{k}^{{\rm {\mathcal C}}} \left(\varphi _{r,i} \right) =1.$ \\ \hline 
$\nu _{k} \left(\varphi _{r,i} \right)$ & Parameter, associated with a particular $\varphi _{r,i} $ and $k$-th class $C_{k} $, $\nu _{k} \left(\varphi _{r,i} \right)={\textstyle\frac{1}{\pi }} \left(\varphi _{r,i} \right).$ \\ \hline 
$\nu _{k} \left(\vec{\varphi }_{r} \right)$ & Collection of $L$ parameters,  $\nu _{k} \left(\vec{\varphi }_{r} \right)=\left[\nu _{k} \left(\varphi _{r,1} \right),\ldots ,\nu _{k} \left(\varphi _{r,L} \right)\right]$, where $\sum _{i=1}^{L}\nu _{k} \left(\varphi _{r,i} \right) =1$. \\ \hline 
$\phi _{k} \left(\vec{\varphi }_{r} \right)$ & Non-linear map in ${\rm {\mathcal H}}$, defined for a stabilized sequence $\vec{\varphi }_{r} $ as  $\phi _{k} \left(\vec{\varphi }_{r} \right)=\left(\nu _{k} \left(\vec{\varphi }_{r} \right)\right)^{T} f_{k}^{{\rm {\mathcal C}}} \left(\vec{\varphi }_{r} \right)$, where $\nu _{k} \left(\varphi _{r,i} \right)={\textstyle\frac{1}{\pi }} \left(\varphi _{r,i} \right),$ and $f_{k}^{{\rm {\mathcal C}}} \left(\varphi _{r,i} \right)\in \left[0,1\right]$ outputs a probability. \\ \hline 
$\iota \left(\cdot \right)$ & Function, returns an inner product. \\ \hline 
$Z$ & A parameter of procedure ${\rm {\mathcal P}}_{S} $. \\ \hline 
$B$ & A parameter of procedure ${\rm {\mathcal P}}_{S} $. \\ \hline 
$\ell _{k} \left(\vec{\varphi }_{r} \right)$ & A parameter of algorithm ${\rm {\mathcal A}}_{C} $. \\ \hline 
$\xi \left(\vec{\varphi }_{r} \right)$ & A parameter of algorithm ${\rm {\mathcal A}}_{C} $. \\ \hline
\end{longtable}
\end{center}
\end{document}